\RequirePackage{etex}
\documentclass[preprint,12pt]{elsarticle}

\usepackage{amssymb}
\usepackage{amsmath}
\usepackage{lineno} 
\usepackage{color}
\usepackage{autonum}
\usepackage{booktabs}

\renewcommand{\b}[1]{{\boldsymbol #1}}
\newcommand{\x}{\mathbf{x}}
\newcommand{\z}{\mathbf{z}}
\newcommand{\cs}{\stackrel{a.s.}{\longrightarrow}}
\newtheorem{theo}{Theorem}[section]
\newtheorem{cor}{Corollary}[section]
\newtheorem{defi}{Definition}[section]
\newtheorem{lemma}{Lemma}[section]
\newtheorem{Remark}{Remark}[section]
\DeclareMathOperator*{\argmin}{\textrm{argmin}}
\DeclareMathOperator*{\logit}{\textrm{logit}}

\begin{document}

\begin{frontmatter}

\title{Estimating the Partially Linear Zero-Inflated Poisson Regression Model: a Robust Approach Using a EM-like Algorithm}

\author[add1,add2]{María José Llop\corref{cor1}}
\ead{llopmariajose@gmail.com}
\cortext[cor1]{Corresponding author}

\author[add1]{Andrea Bergesio}
\author[add3]{Anne-Fran\c{c}oise Yao}

\address[add1]{Departamento de Matemática, Facultad de Ingeniería Química, Universidad Nacional del Litoral, Argentina}
\address[add2]{CONICET, Argentina}
\address[add3]{Laboratoire de Mathématiques, Université Clermont Auvergne, France}

\begin{abstract}

Count data with an excessive number of zeros frequently arise in fields such as economics, medicine, and public health. Traditional count models often fail to adequately handle such data, especially when the relationship between the response and some predictors is nonlinear. To overcome these limitations, the partially linear zero-inflated Poisson (PLZIP) model has been proposed as a flexible alternative. However, all existing estimation approaches for this model are based on likelihood, which is known to be highly sensitive to outliers and slight deviations from the model assumptions. This article presents the first robust estimation method specifically developed for the PLZIP model. An Expectation-Maximization-like algorithm is used to take advantage of the mixture nature of the model and to address extreme observations in both the response and the covariates. Results of the algorithm convergence and the consistency of the estimators are proved. A simulation study under various contamination schemes showed the robustness and efficiency of the proposed estimators in finite samples, compared to classical estimators. Finally, the application of the methodology is illustrated through an example using real data.
\end{abstract}


\end{frontmatter}

\section{Introduction}\label{intro} 

Count data frequently arise in fields such as economics, medicine, and public health. These data often exhibit more zeros than standard count models like Poisson or negative binomial can arise. Zero-inflated models address this issue by modeling the data as a two-component mixture: a degenerate distribution with mass one at zero and a non-degenerate count distribution. Thus, a binomial component determines whether a count arises from structural zeros or from the count process. To analyze such data, \cite{Lambert92} introduced the zero-inflated Poisson (ZIP) model. However, ZIP regression can be too restrictive when the relationship between the response and some predictors is nonlinear. To overcome this, \cite{Lam06} proposed the partially linear ZIP model (PLZIP), which incorporates a nonparametric component to relax model assumptions and better capture complex associations. Specifically, we will say throughout the paper that a random variable $y\in \mathbb{N}_0$ follows the PLZIP model if 
\begin{equation}
y \sim
    \begin{cases} 
      0 & \textup{with probability } \pi \\
      \mathcal{P}(\lambda) & \textup{with probability } 1-\pi
   \end{cases}
\end{equation}
where $0\leq \pi \leq 1$ and $\mathcal{P}(\lambda)$ is a Poisson distribution with parameter $\lambda>0$. 
The mixing probability $\pi$ is modeled via a logistic link as $\logit(\pi)= \log\left( \frac{\pi}{1-\pi}\right)=\z^\top\b{\gamma}$ and the Poisson mean is modeled as $\log(\lambda)=\x^\top\b{\beta}+m(t)$, where $\z\in \mathbb{R}^q$, $\x\in \mathbb{R}^p$ and $t \in \mathbb{R}$ are predictor variables, while $\b{\beta}\in \mathbb{R}^p$ and $\b{\gamma} \in \mathbb{R}^q$ are the regression parameters, and $m$ is a continuous function. So, the parameter space for the PLZIP model is defined as
\begin{align}
    \Theta = \lbrace \b{\theta} = (\b{\beta},\b{\gamma},m) \in \mathbb{R}^p \times \mathbb{R}^q \times \mathcal{C} \rbrace,
\end{align}
where $\mathcal{C}$ is the space of continuous functions. 

The semiparametric structure of the PLZIP model requires both classical estimation methods for the parameters, such as maximum likelihood, and smoothing techniques for the nonparametric part. In this context, \cite{Lam06} proposed sieve maximum likelihood estimators using piecewise linear approximations. Later extensions include the doubly semiparametric ZIP model by \cite{He10}, and the single-index ZIP model estimated via profile likelihood with B-splines by \cite{Wang15}. However, all these methods rely on likelihood-based estimation, which is known to be highly sensitive to outliers or slight deviations from model assumptions.

Classical M-estimators, introduced by \cite{Huber64}, provide robust alternatives to maximum likelihood by replacing the score function with a robust estimating function that downweights observations with large residuals. However, in mixture models (particularly with well-separated components) many observations may naturally have large residuals relative to the overall mean, leading to excessive downweighting and biased estimates. To address this limitation, \cite{Hall10} proposed the robust expectation solution (RES) algorithm, which uses the Expectation-Maximization (EM) algorithm to identify the component membership of each observation and reduce the influence of extreme observations within each component.

In partially linear models, extreme values can also affect the shape and scale of the estimate of the non-parametric component. In this context, \cite{Boente06} analyzed the sensitivity of classical estimators in generalized partially linear models and proposed robust alternatives based on likelihood profiles. Later, \cite{Boente10} 
introduced a computationally efficient three-step robust procedure that estimates the linear and nonparametric components simultaneously, solving the dependency issue in profile likelihoods. They also established consistency, asymptotic normality, and developed a robust test for the parametric component.

The main goal of this article is to introduce a new robust estimation procedure for the PLZIP regression model. To achieve this, we propose an EM-like algorithm combined with the three-step procedure introduced by \cite{Boente10}, which enables the sequential estimation of linear and nonparametric components. Robust loss functions are incorporated into the M-step of the EM framework, instead of traditional log-likelihood functions, reducing the contribution of extreme observations in terms of the component they belong to. Taking advantage of the mixed nature of the model, this approach addresses both the potential complex relationships between variables and sensitivity to outliers. The proposal and implementation of a robust method specifically developed for the PLZIP model becomes a useful tool for analyzing zero-inflated count data with complex predictor relationships.

The paper is organized as follows: In Section \ref{proposal} we introduce the robust estimator for the PLZIP model. In Section \ref{algorithm} a robust estimation EM-like algorithm is introduced to obtain the proposed estimators and a result on the convergence of the algorithm is given. The consistency of the obtained estimators is proved in Section \ref{asymptotics}. In Section \ref{simulations} we present simulation studies to compare the performance of the estimators under different contamination schemes. Finally, the application of the methodology is illustrated through an example using healthcare data in Section \ref{real_data}.

\section{Robust estimators for the PLZIP model}\label{proposal}

Let $\lbrace (y_i,\x_i^\top,\z_i^\top,t_i)\rbrace_{i=1}^n \in \mathbb{N}_0\times \mathbb{R}^p\times\mathbb{R}^q\times\mathbb{R}$ a random sample that follows the PLZIP model. Defining $\mathbf{y}=(y_1, \ldots, y_n)$ the response vector, the log-likelihood function for the PLZIP model is given by 
\begin{align}\label{PLZIP_empirical_log-likelihood}
l_n(\mathbf{y},\b{\theta})&= \sum_{y_i=0}\log\left(e^{\z_i^\top\b{\gamma}}+\exp{(-e^{\x_i^\top\b{\beta} + m(t_i)}})\right)-\sum_{i=1}^n\log(1+e^{\z_i^\top\b{\gamma}})\\
\nonumber
&+ \sum_{y_i>0}\left({y_i(\x_i^\top\b{\beta}+ m(t_i)) -e^{\x_i^\top\b{\beta}+m(t_i)} }\right)-\sum_{y_i>0}\log\left({y_i!}\right).
\end{align}

Let us suppose that we can observe an auxiliary (also known as latent or missing) variable $w$ that indicates whether the response $y$ arises from the degenerate distribution, i.e. it is a structural zero ($w = 1$), or from the Poisson distribution ($w = 0$). 
Then, the complete-data log-likelihood function for the complete data $(\mathbf{y},\mathbf{w})$ where $\mathbf{w}=(w_1, \ldots, w_n)$, results 
\begin{align}\label{complete_log_likelihood}
l_n^c(\mathbf{y},\mathbf{w},\b{\theta})&=\sum_{i=1}^n(1-w_i)\left(y_i\left(\x_i^\top\b{\beta} + m(t_i)\right) - e^{\x_i^\top\b{\beta + m(t_i)}} - \log\left({y_i!}\right)\right)\\
&+ 
\sum_{i=1}^n\left(w_i\z_i^\top\b{\gamma}-\log\left(1+e^{\z_i^\top\b{\gamma}}\right)\right).
\end{align}
Note that the first summation depends only on $\b{\beta}$ and $m$, while the second summation depends only on $\b{\gamma}$. This separability allows the two terms to be optimized independently, simplifying the estimation procedure. 

The fact that the first term in \eqref{complete_log_likelihood} involves both a linear component $\b{\beta}$ and a nonparametric component $m$, makes it necessary to combine parametric and nonparametric estimation techniques. To address this, we adopt the estimation procedure proposed by \cite{Boente10}, which allows these components to be estimated sequentially. Specifically, we proceed through the following three steps:

\underline{Step 1}: For each fixed value of the variable $t$, let say $t=\tau$, denote $\eta = m(\tau)$ and compute
    \begin{equation}
        (\Tilde{\b{\beta}}_n(\tau),\Tilde{m}_n(\tau)) = \argmin_{\b{\beta},\
        \eta } \sum_{i=1}^n W_{i,n}(\tau) (1-w_i)\left(e^{\x_i^\top\b{\beta} + \eta} -y_i\left(\x_i^\top\b{\beta} + \eta \right) + \log\left({y_i!}\right)\right),
    \end{equation}
where $W_{i,n}(\tau)$ is the classical Nadaraya-Watson kernel type weights
    \begin{equation}\label{weights}
        W_{i,n}(\tau)=\frac{K\left(\frac{\tau-t_i}{h_n}\right)}{\sum_{j=1}^n  K\left(\frac{\tau-t_j}{h_n}\right)},
    \end{equation}
with $K$ a kernel function usually symmetric and positive, and $h_n>0$ a smoothing parameter selected to match the smoothness of $m$. Following the weighted likelihood estimation approach proposed by \cite{Staniswalis89} and \cite{Severini94}, kernel weights are incorporated into the estimation procedure to exploit the assumed smoothness of the function $m$. Specifically, the assumption that $m$ is smooth implies that if $t_j$ is close to $\tau$, then hen the corresponding values $y_j$ and $\x_j$ provide relevant information about $m(\tau)$. 
    
\underline{Step 2}: Once an initial estimate of the nonparametric component has been obtained, it is then used to get an estimate (independent of $\tau$) of the linear parameter, as follows
    \begin{equation}
        \widehat{\b{\beta}}_n = \argmin_{\b{\beta}} \sum_{i=1}^n(1-w_i)\left(e^{\x_i^\top\b{\beta} + \Tilde{m}_n(t_i)} - y_i\left(\x_i^\top\b{\beta} + \Tilde{m}_n(t_i)\right) + \log\left({y_i!}\right)\right)
    \end{equation}
and an estimator of $\b{\gamma}$ is found as
    \begin{equation}\label{gamma_hat}
        \widehat{\b{\gamma}}_n = \argmin_\b{\gamma} \sum_{i=1}^n\left(\log\left(1+e^{\z_i^\top\b{\gamma}}\right) - w_i^{(r)}\z_i^\top\b{\gamma} \right).
    \end{equation}
    
\underline{Step 3}: The estimate of $\b{\beta}$ obtained in Step 2 is then used to improve the nonparametric estimation of $m$
    \begin{equation}
        \widehat{m}_n(\tau) = \argmin_{\eta} \sum_{i=1}^n(1-w_i)\left(e^{\x_i^\top\widehat{\b{\beta}} + \eta} - y_i\left(\x_i^\top\widehat{\b{\beta}} + \eta \right) + \log\left({y_i!}\right)\right),
    \end{equation}
where, again, this estimation is carried out for each fixed value $t = \tau$.

\

The procedure described above relies on the assumption that the auxiliary variable $w_i$ is observable, however, it is not possible to determine whether a zero is structural or comes from the Poisson component. In this context we use the main idea of the EM algorithm, which proceeds by computing the conditional expectation of \eqref{complete_log_likelihood} given the response and then maximizing it. Furthermore, because the complete-data log-likelihood is linear in $w_i$, this step is equivalent to directly replacing $w_i$ by its conditional expectation $E(w_i | y_i)$ in \eqref{complete_log_likelihood}. In that direction, by Bayes' rule, it is easy to show that (see \cite{Lambert92}, page 4) 
\begin{align}\label{conditional_expectation}
E_{\b{\theta}}(w_i|y_i)&=\begin{cases}
\dfrac{1}{1+\exp\{-\z_i^\top\b{\gamma}-e^{\x_i^\top\b{\beta}+m(t_i)}\}}&\text{ if  }y_i=0\\
0&\text{ if  }y_i>0
\end{cases}.
\end{align}
\

To obtain robust estimators, we propose replacing the estimation functions used in the three steps with robust alternatives. In that direction, consider $\rho: \mathbb{N}_0 \times \mathbb{R} \rightarrow \mathbb{R}$ a loss function (on which conditions are imposed in Section \ref{asymptotics}), $\omega_1: \mathbb{R}^p \rightarrow \mathbb{R}$ and $\omega_2: \mathbb{R}^q \rightarrow \mathbb{R}$ penalty functions controlling the high leverage $\x$ and $\z$ predictors, respectively. Define the following functionals
\begin{align}
    \label{Q_1}
    Q_{1}(\eta,\b{\beta},\tau|\b{\theta}) &= E_0\left[ (1-E_{\b{\theta}}(w|y))\rho(y, \x^\top\b{\beta} + \eta)\omega_1(\x)|t=\tau \right], \\
    \label{Q_2}
    Q_{2}(\b{\beta},m|\b{\theta}) &= E_0 \left[ (1-E_{\b{\theta}}(w|y))\rho(y, \x^\top\b{\beta} + m(t))\omega_1(\x) \right], \\
    \label{Q_3}
    Q_{3}(\b{\gamma}|\b{\theta}) &= E_0 \left[ \left(\log\left(1+e^{\z^\top\b{\gamma}}\right)-E_{\b{\theta}}(w|y) \z^\top\b{\gamma}\right)\omega_2(\z) \right],
\end{align}
where $E_0$ denote the expectation under the true model, this is, when the true parameter value is $\b{\theta}_0 = (\b{\beta}_0,\b{\gamma}_0,m_0) \in \Theta$. The last argument $\b{\theta}$ emphasizes that the conditional expectation $E_{\b{\theta}}(w|y)$ is taken under the parameter $\b{\theta}$.

Throughout this article we assume that $\rho$, $\omega_1$ and $\omega_2$ are such that 
\begin{align}
    (\Tilde{\b{\beta}}_0(\tau), m_0(\tau)) &= \argmin_{\b{\beta},\eta} Q_{1}(\eta,\b{\beta},\tau|\b{\theta}),\\
    \b{\beta}_0 &= \argmin_{\b{\beta}} Q_{2}(\b{\beta},m_0|\b{\theta}), \\
    \b{\gamma}_0 &= \argmin_{\b{\gamma}} Q_{3}(\b{\gamma}|\b{\theta}),
\end{align} 
and $\rho$ is selected to ensure the uniqueness of $\b{\beta}_0$ and $m_0(\cdot)$.

As motivated by the discussion above regarding the smoothness of $m$ and the use of kernel weights, we construct the empirical analogue of \eqref{Q_1} by means of a robust version of the weighted likelihood approach proposed by \cite{Staniswalis89} and \cite{Severini94}. Specifically,
\begin{equation}\label{Q_n1}
    Q_{n,1}(\eta,\b{\beta},\tau|\b{\theta}) = \sum_{i=1}^n W_{i,n}(\tau)(1-E_{\b{\theta}}(w_i|y_i))\rho(y_i, \x_i^\top\b{\beta} + \eta)\omega_1(\x_i).
\end{equation}
To obtain the empirical analogues of \eqref{Q_2} and \eqref{Q_3}, we replace the expectation operator $E_0$ by integration with respect to the empirical distribution of the observed data, which yields the following expressions
\begin{align}
    \label{Q_n2}
    Q_{n,2}(m,\b{\beta}|\b{\theta}) &= \frac{1}{n}\sum_{i=1}^n (1-E_{\b{\theta}}(w_i|y_i))\rho(y_i, \x_i^\top\b{\beta} + m(t_i))\omega_1(\x_i),\\
    \label{Q_n3}
    Q_{n,3}(\b{\gamma}|\b{\theta}) &= \frac{1}{n} \sum_{i=1}^n\left(\log\left(1+e^{\z_i^\top\b{\gamma}}\right)-E_{\b{\theta}}(w_i|y_i)\z_i^\top\b{\gamma}\right)\omega_2(\z_i).
\end{align}

\begin{defi}\label{estimator_definition}
    Given the random sample $\lbrace(y_i,\x_i^\top,\z_i^\top,t_i)\rbrace_{i=1}^n \in \mathbb{N}_0\times \mathbb{R}^p\times\mathbb{R}^q\times\mathbb{R}$, which follows the PLZIP model with parameter $\b{\theta}_0=(\b{\beta}_0,\b{\gamma}_0,m_0) \in \Theta$,
    the \textit{robust three-step estimator} $\widehat{\b{\theta}}_n=(\widehat{\b{\beta}}_n,\widehat{\b{\gamma}}_n,\widehat{m}_n)$ for $\b{\theta}_0$ is given by the following steps: 
\begin{itemize}
    \item[] \underline{Step 1}:
    \begin{align}
        (\Tilde{\b{\beta}}_n(\tau),\Tilde{m}_n(\tau)) = \argmin_{\b{\beta},
        \eta } Q_{n,1}(\eta,\b{\beta},\tau| \b{\theta}).
    \end{align}
    \item[] \underline{Step 2}:
    \begin{align}
        \widehat{\b{\beta}}_n &= \argmin_{\b{\beta}} Q_{n,2}(\Tilde{m}_n,\b{\beta}|\b{\theta}), \\
        \widehat{\b{\gamma}}_n &= \argmin_{\b{\gamma}} Q_{n,3}(\b{\gamma}| \b{\theta}).
    \end{align}
    \item[] \underline{Step 3}:
    \begin{align}
        \widehat{m}_n(\tau) = \argmin_{\eta} Q_{n,1}(\eta,\widehat{\b{\beta}}_n,\tau|\b{\theta}).
    \end{align}
\end{itemize}
\end{defi}

The estimators also can be obtained by differentiating the functions $Q_{n,1}$, $Q_{n,2}$ and $Q_{n,3}$ and equating them to zero. More specifically, let 
\begin{align}\label{S_n}
    S_n(\b{\theta})= (S_{n,1}(\eta,\b{\beta},\tau| \b{\theta})^\top,S_{n,2}(m,\b{\beta}| \b{\theta})^\top,S_{n,3}(\b{\gamma}| \b{\theta})^\top)^\top
\end{align}
with
\begin{align}
    \nonumber
    S_{n,1}(\eta,\b{\beta},\tau| \b{\theta})
    &= \sum_{i=1}^n W_{i,n}(\tau)(1-E_{\b{\theta}}(w_i|y_i))\Psi(y_i, \x_i^\top\b{\beta} + \eta)\omega_1(\x_i)\b{a}_i,\\
     S_{n,2}(m,\b{\beta}| \b{\theta}) &= \frac{1}{n}\sum_{i=1}^n(1-E_{\b{\theta}}(w_i|y_i))\Psi(y_i, \x_i^\top\b{\beta} + m(t_i))\omega_1(\x_i) \x_i^\top, \\
     \nonumber
     S_{n,3}(\b{\gamma}| \b{\theta}) &= \frac{1}{n} \sum_{i=1}^n\left(\left(1+e^{-\z_i^\top\b{\gamma}}\right)^{-1} -E_{\b{\theta}}(w_i|y_i)  \right) \z_i^\top  \omega_2(\z_i),
\end{align}
where $\Psi(y,u) = \frac{\partial}{\partial u}\rho(y,u)$ and $\b{a}_i=(1,\x_i)^\top$. 
Then, the \textit{empirical estimation equations system} giving rise to the estimators is
\begin{align}\label{equations}
    S_n(\b{\theta}) =\b{0}.
\end{align}

Under the true model, e.i. when the parameter value is $\b{\theta}_0$, \eqref{S_n} is given by 
\begin{align}
    S_0(\b{\theta})= (S_{1}(\b{\beta}, \eta,\tau|\b{\theta})^\top, S_{2}(\b{\beta},m|\b{\theta})^\top, S_{3}(\b{\gamma}|\b{\theta})^\top ),
\end{align} 
with
\begin{align}
    S_{1}(\b{\beta}, \eta,\tau|\b{\theta})
    &= E_0[(1-E_{\b{\theta}}(w|y))\Psi(y, \x^\top\b{\beta} + \eta)\omega_1(\x)\b{a}|t=\tau],\\
     S_{2}(\b{\beta},m|\b{\theta}) &= E_0[(1-E_{\b{\theta}}(w|y))\Psi(y, \x^\top\b{\beta} + m(t))\omega_1(\x) \x^\top], \\
    S_{3}(\b{\gamma}|\b{\theta}) &=E_0\left[ \left(\left(1+e^{-\z^\top\b{\gamma}}\right)^{-1} -E_{\b{\theta}}(w|y)  \right) \z^\top  \omega_2(\z)\right] ,
\end{align}
where $\b{a} = (1,\x^\top)^\top$. Then, the theoretical system of estimation equations results
\begin{align}\label{theorical_equations}
    S_0(\b{\theta}) =\b{0}.
\end{align} 

Note that if we choose the function $\rho$ is derivable with derivative $\Psi$ satisfying $E_{\b{\theta}}(\Psi(y_i, \x^\top\b{\beta} + m(t))| \x,t)=0$, it is easy to show that $\b{\theta}_0$ verifies \eqref{theorical_equations}. Condition $E_{\b{\theta}}(\Psi(y_i, \x^\top\b{\beta} + m(t))| \x,t)=0$ is known as \textit{conditional Fisher-consistency}. To satisfy this condition, the loss functions $\rho$ considered in this work are defined as follows
\begin{align}\label{rho}
    \rho(y,u) = \phi(e^u -y(u +1 - \ln{y})) + G(e^u)
\end{align}
where $\phi(x)$ is a bounded non-decreasing function for $x>0$, with continuous derivative $\varphi$, and $G(e^u)$ is a correction term that guarantees the Fisher consistency condition. In this case, it is given by
\begin{align}\label{G}
    G(s) &= -\int\varphi (s)e^{-s} ds \\
    &+ \sum_{j=1}^\infty \int \varphi (s -j(ln(s) +1 - \ln{j}))\frac{\left(s^j e^{-s}\right)}{j!}\frac{(j-s)}{s}ds.
\end{align}

\begin{Remark}\label{r_gamma}
    A robust loss function is not needed in \eqref{Q_n3} since, as discussed by \cite{Hall10}, reducing the influence of high-leverage covariate observations by means of the weight function $\omega_2$ and ensuring an accurate estimation of $w_i$ through its conditional expectation (obtained via a robust estimate of $\b{\beta}$) is sufficient to preserve the robustness of the estimator $\b{\gamma}$.
\end{Remark}

\section{The robust EM-like algorithm}\label{algorithm}

To compute the estimator proposed in Definition \ref{estimator_definition}, we introduce an EM-like algorithm that combines the three-step procedure described above with the EM framework. It consists of the following steps:
\begin{itemize}
    \item[] (Initialization) Get an initial value $\widehat{\b{\theta}}_n^{(0)}=(\widehat{\b{\beta}}_n^{(0)}, \widehat{\b{\gamma}}_n^{(0)}, \widehat{m}_n^{(0)})$.
    \item[] ($E$ step) In the $r$-th iteration estimate $w_i$ by $w_i^{(r)} = E_{\widehat{\b{\theta}}_n^{(r)}}(w_i|y_i)$ where the expectation is taken under the current parameter estimator $\widehat{\b{\theta}}_n^{(r)}$.
    \item[] ($M$ step) Carry out the three-step procedure:
    \begin{itemize}
        \item[] \underline{Step 1}: For each fixed value $t=\tau$ denote $\eta = m(\tau)$ and calculate 
            \begin{equation}\label{est_tilde}
            (\Tilde{\b{\beta}}_n(\tau),\Tilde{m}_n(\tau)) = \argmin_{\b{\beta},
            \eta } Q_{n,1}(\eta,\b{\beta},\tau|\widehat{\b{\theta}}_n^{(r)}).
        \end{equation}
        \item[] \underline{Step 2}: Compute $\b{\beta}$ estimator using $\Tilde{m}_n$ from Step 1
        \begin{equation}\label{beta_hat}
            \widehat{\b{\beta}}_n^{(r+1)} = \argmin_{\b{\beta}} Q_{n,2}(\Tilde{m}_n,\b{\beta}|\widehat{\b{\theta}}_n^{(r)})
        \end{equation}
        and find $\b{\gamma}$ estimate as 
        \begin{equation}\label{gamma_hat}
            \widehat{\b{\gamma}}_n^{(r+1)} = \argmin_\b{\gamma} Q_{n,3}(\b{\gamma}|\widehat{\b{\theta}}_n^{(r)}).
        \end{equation}
        \item[] \underline{Step 3}: Using $\widehat{\b{\beta}}_n^{(r+1)}$ from Step 2, update $m$ estimator as
        \begin{equation}
            \widehat{m}_n(\tau) = \argmin_{\eta} Q_{n,1}(\eta,\widehat{\b{\beta}}_n^{(r+1)},\tau|\widehat{\b{\theta}}_n^{(r)}).
        \end{equation}
    \end{itemize}
\end{itemize}

\begin{Remark}
Adapting the profile likelihood approach used by \cite{Boente06}, for each fixed value of $\b{\beta} \in \mathbb{R}^p$ we should calculate $\Tilde{m}_{\b{\beta},n}(\tau) = \argmin_{\eta} Q_{n,1}(\eta, \b{\beta}, \tau |\widehat{\b{\theta}}_n)$. Consequently, it would be necessary to compute $\Tilde{m}_{\b{\beta},n}$ at each data point $t_i$ over a grid of possible values of $\b{\beta}$, in order to calculate the objective function $Q_{n,2}(\b{\beta}, \Tilde{m}_{\b{\beta},n}| \widehat{\b{\theta}}_n)$, which results in a considerable computational cost. In contrast, the simultaneous estimation of $\b{\beta}$ and $\eta$ in Step 1 avoids nested optimization, since $\b{\beta}$ and $\eta$ are estimated jointly for each data point $t_i$, without the need to do so for a grid of $\b{\beta}$ values, reducing the computational cost.
\end{Remark}

\subsection{Algorithm convergence}\label{algorithm_convergence}

In \cite{Hall10} the authors proposed robust estimating equations to obtain the ZIP model estimators by the expectation-solution (ES) algorithm. This is similar to the EM but an estimating equation is solved instead of maximizing the expected log-likelihood. When the derivatives of the log-likelihood function give the equations of the ES algorithm, the ES and EM algorithms are equivalent. \cite{Rosen00} showed under certain regularity conditions that if the ES algorithm converges, it converges to a solution of an unbiased estimating function. In this article, a similar result is proved for the estimation equations defined in \eqref{equations}.

\begin{theo}\label{theorem_algorithm}
    If there exists a point $\b{\theta}^*$ such that the ES algorithm converges to, this is $\lim_{r \rightarrow \infty} \widehat{\b{\theta}}_n^{(r)} = \b{\theta}^*$, and $\rho$ is such that its derivative $\Psi$ is continuous and verifies $E_{\b{\theta}}(\Psi(y, \x^\top\b{\beta} + m(t))|\x,t)=0$, then $S_n(\b{\theta}^*)=0$. Furthermore, $S_n$ is an unbiased estimation equation, this is $E_{\b{\theta}}(S_n(\b{\theta}))=0$ for each $\b{\theta} \in \Theta$.
\end{theo}

The proof of Theorem \ref{theorem_algorithm} follows from the Proposition in Appendix of \cite{Rosen00}, since it is easy to show that the following conditions are satisfied
\begin{enumerate}[1)]
    \item $E_{\b{\theta}}(S_{n,1}(\b{\beta},\eta,\tau|\b{\theta}))=\b{
0}$, $E_{\b{\theta}}(S_{n,2}(\b{\beta},m|\b{\theta}))=\b{
0}$ and $E_{\b{\theta}}(S_{n,3}(\b{\gamma}|\b{\theta}))=\b{
0}$.
    \item $S_n(\b{\theta})$ is a continuous function in $\Theta$.
\end{enumerate}

\begin{Remark}
Theorem \ref{algorithm_convergence} guarantees that the algorithm converges to a solution of the system $S_n(\b{\theta}) = 0$. Moreover, the robust three-step procedure estimator introduced in Definition \ref{estimator_definition} is also a solution to this system; however, uniqueness is not guaranteed. Consequently, a good initialization of the algorithm is essential in practice to reaching the global minimum and to avoid convergence to local solutions.
\end{Remark}

\section{Asymptotic properties}\label{asymptotics}

Under the regularity conditions presented in \cite{Boente10} we prove the consistency of the estimators presented in this work.
\begin{enumerate}[C1]
    \item $\rho(y,u)$ is continuously differentiable and bounded; $\omega_1$ and $\omega_2$ are bounded.
    \item The kernel function $K: \mathbb{R} \rightarrow \mathbb{R}$ is even, non-negative, bounded and continuous with a bounded derivative, which satisfies $\int K(u)du = 1$, $\int u^2K(u)du<\infty$ and $|u|K(u) \rightarrow 0$ if $|u|\rightarrow 0$.
    \item The sequence $h_n$ satisfies $h_n \rightarrow 0$ y $n h_n /\log(n) \rightarrow \infty$.
    \item The marginal density of $t$, $f_T$, is a bounded function on $\mathcal{T}$. Let $\mathcal{T}_0 \subset \mathcal{T}$ be compact set and let $\mathcal{T}_\delta \subset \mathcal{T}$ be the closure of a $\delta$-neighborhood of $\mathcal{T}_0$, there exists a positive constant $A_1(\mathcal{T}_\delta)$ such that $A_1(\mathcal{T}_\delta) < f_T(t) \ \forall \ t \in \mathcal{T}_\delta$.
    \item $Q_1(\eta,\b{\beta},\tau|\b{\theta})$ satisfies the following equicontinuity condition: given $\varepsilon>0 \ \exists \ \delta >0$ such that $|\tau_1-\tau_2|<\delta$ for $\tau_i \in \mathcal{T}_0$ and $||\b{\beta}_1 - \b{\beta}_2||<\delta$ for $\b{\beta}_i \in \mathcal{K}$, $i=1,2$. Then $\sup_{\eta \in \mathbb{R}} |Q_1(\eta,\b{\beta}_1,\tau_1|\b{\theta}) - Q_1(\eta,\b{\beta}_2,\tau_2|\b{\theta})| < \varepsilon$. $||\cdot||$ indicates $L2-$norm.
    \item $Q_1(\eta,\b{\beta},\tau|\b{\theta})$ and $m_0(t)$ are continuous.
\end{enumerate}

To introduce the notation to be used from now on, suppose that $v:\mathcal{T} \rightarrow \mathbb{R}$ is a continuous function and define $||v||_{\infty} = \sup_{t\in \mathcal{T}} |v(t)|$ and $||v||_{0,\infty} = \sup_{t\in \mathcal{T}_0} |v(t)|$ where $\mathcal{T}_0 \subset \mathcal{T}$ is a compact set. $\mathcal{N}(\varepsilon, \mathcal{F}, L^1(Q))$ denotes the covering number of a family of functions $\mathcal{F}$. Specific details about the covering number of a class can be found in \cite{Pollard84}.

\begin{lemma}\label{lemma}
    Let $\mathcal{K} \subset \mathbb{R}^p$ be a compact set, suppose that C1-C6 are satisfied and $K$ is a bounded variation function. Consider the family of functions
    \begin{equation}
        \mathcal{F} = \{ f_{\b{\beta},\eta}(y, \x|\b{\theta}) = (1-E_\b{\theta}(w|y))\rho(y, \x^\top\b{\beta} + \eta)\omega_1(\x), \ \b{\beta} \in \mathcal{K}, \ \eta \in \mathbb{R}\}.
    \end{equation}
    Suppose that exist constants $A$ and $W$ such that $\mathcal{F}$ has a covering number
    \begin{equation}
        \sup_{Q}\mathcal{N}(\varepsilon, \mathcal{F}, L^1(Q)) \leq A \varepsilon^{-W}.
    \end{equation}
    Then,
    \begin{enumerate}[(a)]
        \item $\sup_{\b{\beta} \in \mathcal{K}, \
       \eta \in \mathbb{R}} ||Q_{n,1}(\eta,\b{\beta},\cdot|\b{\theta}) - Q_1(\eta,\b{\beta},\cdot|\b{\theta})||_{0,\infty} \cs 0$
       \item $||\Tilde{m}_n - m_0||_{0,\infty} \cs 0$ and $||\Tilde{\b{\beta}}_n- \b{\beta}_0||_{0,\infty} \cs 0$.
    \end{enumerate}
\end{lemma}

\textbf{Proof of Lemma \ref{lemma}.}

The proof of (a) follows from the steps of the proof of Lemma 3.1 of \cite{Boente06}, by bounding the term $|1-E_\b{\theta}(w_i|y_i)| < 1$. See details in \ref{appendix}. (b) is easy to prove using (a), $C6$ and Lemma A1 of \cite{Carroll97}. In the notation of the Lemma A1, take $f=Q_1$ with unique minimizer $ (\b{\beta}_0, m_0(t)) \in \mathcal{K} \times \mathcal{C}$ and $f_n = Q_{n,1}$ with a minimizer $(\Tilde{\b{\beta}}_n(t), \Tilde{m}_n(t))$. 

$\hfill\blacksquare$

\begin{theo}\label{theorem_beta}     
     Suppose that $\Psi$ is bounded, conditions of Lemma \ref{lemma} are satisfied and the family of functions 
    \begin{align}
        \mathcal{H} = \{ h_{\b{\beta}}(y, \x, t|\b{\theta}) = (1-E_\b{\theta}(w|y))\rho(y, \x^\top\b{\beta} + \Tilde{m}_n(t))\omega_1(\x), \ \b{\beta} \in \mathbb{R}^{p}\}
    \end{align} 
    has a coverage number such that $\log \mathcal{N}(\varepsilon, \mathcal{H}, L^1(P_n)) = o_p(n)$, where $P_n$ is the empirical distribution of $(y,\x,t)$. Then the following conditions are satisfied,
     \begin{enumerate}[(a)]
         \item $sup_{\b{\beta} \in \mathbb{R}^{p}} |Q_{n,2}( \Tilde{m}_n,\b{\beta}|\b{\theta})-Q_{2}(m_0,\b{\beta}|\b{\theta})| \cs 0.$ 
         \item If $Q_{2}(m_0,\b{\beta}|\b{\theta})$ has a unique minimum $\b{\beta}_0$, $\x$ has a second finite moment and $lim \inf_{||\b{\beta}|| \rightarrow \infty} Q_{2}(m_0,\b{\beta}|\b{\theta}) > Q_{2}(m_0,\b{\beta}_0|\b{\theta})$, then $\widehat{\b{\beta}}_n \cs \b{\beta}_0.$
     \end{enumerate}
\end{theo}

Proof of Theorem \ref{theorem_beta} is given in \ref{appendix}.

\begin{Remark}
    Applying known results it is easy to see that the condition on the covering number of the family $\mathcal{F}$ in Lemma \ref{lemma} is fulfilled under general conditions when $\rho$ is defined as \eqref{rho}. Indeed, if $\phi$ and $G(e^u)$ of \eqref{rho} have bounded variation, by Lemma 22 of \cite{Nolan87} the covering number of the class $\lbrace G(e^{\x^\top\b{\beta} + \eta}),\ \b{\beta} \in \mathbb{R}^q, \ \eta \in \mathbb{R} \rbrace$ grows at a polynomial rate (is bonded by $A e^{-W}$ for some constants $A$ and $W$). On the other hand, the covering number of the class $\lbrace e^{\x^\top\b{\beta} + \eta} -y(\x^\top\b{\beta} + \eta +1 - \ln{y}), \ \b{\beta} \in \mathbb{R}^p, \ \eta \in \mathbb{R} \rbrace$ grows at a polynomial rate since the family $\lbrace \x^\top \b{\beta} + \eta, \b{\beta} \in \mathbb{R}^p, \ \eta \in \mathbb{R} \rbrace$ has finite dimension (and by Lemma 28 of \cite{Pollard84} grows at a polynomial rate). Using the Lemma 16 of \cite{Nolan87} which states that for $r\geq 1 $ $\mathcal{N}(\varepsilon,\mathcal{H}_1+\mathcal{H}_2, L^r(P)) \leq \mathcal{N}(\varepsilon,\mathcal{H}_1, L^r(P))\mathcal{N}(\varepsilon,\mathcal{H}_2, L^r(P))$ the result is derived. A similar reasoning can be applied to the condition about the covering number of the functions family $\mathcal{H}$ in Theorem \eqref{theorem_beta}. 
\end{Remark}

\begin{cor}\label{corollary_m}
       If hypothesis of Theorem \ref{theorem_beta} are verified, $||\widehat{m}_n - m_0||_{0,\infty} \cs 0$.
\end{cor}
The proof of Corollary \ref{corollary_m} follows easily from Lemma \ref{lemma} and Theorem \ref{theorem_beta}, applying Lemma A1 of \cite{Carroll97}.

\

The following theorem states the consistency of $\widehat{\b{\gamma}}_n$ and is based on an analogous result to that of \cite{Huber67} for the consistency of ML estimator under weaker conditions than usual.

\begin{theo}\label{theorem_gamma}
    Suppose that exist $N_0$ such that $\widehat{\b{\gamma}}_n \in \mathcal{K}_1 \ \forall \ n > N_0$ with $\mathcal{K}_1$ a compact set. Suppose also that the following conditions hold
    \begin{enumerate}[(a)]
        \item Let $\mathcal{B}$ be an open neighborhood of $\b{\gamma}'$,
    \begin{align}
        E_{0} \left[ \inf_{\b{\gamma}\in \mathcal{B}}\left(\log\left(1+e^{\z^\top\b{\gamma}}\right)-E_{\b{\theta}}(w|y) \z^\top\b{\gamma}\right)\omega_2(\z) \right] 
        \rightarrow Q_{3}(\b{\gamma}'|\b{\theta})
    \end{align}
    when $\mathcal{B} \rightarrow \b{\gamma}'$.
    \item Exists $\b{\gamma}_0$ such that $Q_{3}(\b{\gamma}_0|\b{\theta})<Q_{3}(\b{\gamma}|\b{\theta})$ for all $\b{\gamma}$.
    \end{enumerate}
Then $\widehat{\b{\gamma}}_n \cs \b{\gamma}_0$.
\end{theo}

\section{Simulation Studies}\label{simulations}

In this section, we conduct a finite-sample simulation study to evaluate the robustness and performance of the estimators proposed in this article, under four different contamination schemes. For each simulated sample, the PLZIP model is estimated using the procedure described in Section \ref{algorithm}. Additionally, estimators proposed by other authors are computed to enable a comparative assessment of performance.

\subsection{Robust Loss Function}\label{loss}

The loss functions employed in this article have the general form given in \eqref{rho}, where the term $G$, which ensures Fisher consistency, is defined in \eqref{G}. Thus, the only remaining choice is the specification of the function $\phi$. In this work, we consider two different choices for $\phi$, resulting in two robust estimators for each linear parameter and the nonparametric function of the PLZIP model.

One of the $\phi$ functions implemented is that proposed by \cite{Croux02}. In their work, the authors extended the study of \cite{Bianco96} to robustly estimate logistic regression models, deriving conditions for the existence of the estimator in finite samples and proposing a corresponding algorithm for its computation. The robust function is given by
\begin{align}
    \phi_{CH}(s) = 
    \begin{cases}
        s e^{-\sqrt{c}}, \quad  |s|\leq c  \\
        e^{-\sqrt{c}}(2(1+\sqrt{c})+c) -2e^{-\sqrt{s}}(1+\sqrt{s}), \quad |s|>c
    \end{cases}
\end{align} 
where the constant $c$ is chosen following the recommendation of \cite{Croux02}, who suggest selecting it as a trade-off between robustness and efficiency: larger values of $c$ improve efficiency at the expense of robustness, and smaller values increase robustness but reduce efficiency. In this work, we set $c=0.5$ as a balanced choice. The estimators obtained using 
this robust loss function, are denoted CH, which refers to the authors Croux and Haesbroeck.

The second $\phi$ function is based on the redescending M-estimators applied to transformed responses (MT estimators) proposed by \cite{Valdora14}. In their work, the authors suggested transforming the response variable to stabilize its variance to an approximately constant value, ensuring appropriate scaling of the loss function used to define the M-estimator. They demonstrated that, in the case of the Poisson distribution, the square root transformation effectively makes the variance of the response variable nearly constant. Specifically, the robust loss function is given by
\begin{align}\label{rho_valdora}
    \rho(y,u) &= \phi_{MT}( \sqrt{y} - f( e^{u})),
\end{align}
with
\begin{align}
        \phi_{MT}(s) &= \begin{cases}
        1-\left( 1-\left(\frac{s}{c}\right)^2\right)^4, \quad |s|\leq c \\
        1, \quad |s|> c
    \end{cases} 
\end{align}
and
\begin{align}\label{f}
        f(\lambda) &=  \argmin_u E_\lambda(\phi_{MT}(\sqrt{y}-u)),
\end{align}
where the expectation is taken under a Poisson distribution with parameter $\lambda$. As the authors indicate, the constant is chosen $c = 2.9$ to have an efficiency between $75\%$ and $90\%$. The estimators obtained using this robust loss function are denoted MT. 

The weight functions $\omega_1$ and $\omega_2$ that we use are those of \cite{Valdora14}, defined as functions of the Mahalanobis distance.

\begin{Remark}
    It is worth noting that the additional term $G$ is not required in \eqref{rho_valdora}, since, due to the choice of $f$, the function naturally satisfies the Fisher consistency condition.
\end{Remark}

For comparison purposes, we also compute estimators with EM-like algorithm proposed in this work, using the standard log-likelihood function; these are denoted ML. Moreover, to enrich the comparative analysis, we implement the approach proposed by \cite{Wang15} for estimating the PLZIP model based on likelihood profiles and B-spline approximations, which we denote by BS-ML (where BS indicates the use of B-splines).
In addition, we consider estimators of the ZIP linear regression model by introducing the covariate $t$ as a linear predictor. In this context, we include the RES estimator of \cite{Hall10} as well as the classical maximum likelihood estimator for the ZIP model, denoted RES and ZIP-ML, respectively. At this point it is important to highlight that ZIP-ML and RES are estimators of a model that is not the true model of the data, which could affect their performance in simulation studies.

\subsection{Contamination Schemes}

To evaluate the robustness of the proposed estimators we consider four different contamination schemes inspired by the simulation study in \cite{Hall10}. The first one is non-contaminated, while the remaining schemes involve $10\%$ contamination either in the response variable, in the covariates, or in both.

The parameters are set to $\b{\beta} = (2, 2)^\top$ and $\b{\gamma} = (-1, 1)^\top$. The covariates are generated as follows: $x_1 = 1$ for the first $n/2$ observations and $x_1 = 0$ otherwise; $x_2 \sim U[0, 1]$; $t \sim U[-2, 2]$; $z_1 \sim U[0, 1]$; and $z_2 \sim N(0, 1)$. The nonparametric component is defined as $m(t) = \sin\left(\frac{\pi}{2}t\right)$. For this parameter setting, $N=100$ data sets of sample size $n=500$ were generated. The contamination schemes considered in this study are the following:

\underline{Scheme $C_0$:} Non-contaminated.

\underline{Scheme $C_1$:} 
$10\%$ of the response variable is contaminated by adding a large fixed value $y_0 = 70$.

\underline{Scheme $C_2$:} 
$10\%$ of the covariate $x_2$ is contaminated, now generated with distribution $U[1, 2]$. In the same observations, the response is set to $y = 0$. This leads to ``false zeros'' since, for these contaminated observations, $\x_i^\top \b{\beta}$ is large and thus the expected Poisson mean $\lambda_i = e^{\x_i^\top \b{\beta} + m(t_i)}$ would suggest large counts. However, zeros are observed instead. For this contamination scheme, it is necessary to also incorporate the weight $\omega_2(\x_i)$ in \eqref{Q_n3} to properly address the presence of false zeros.

\underline{Scheme $C_3$:} This scenario combines schemes $C_1$ and $C_2$. 
Specifically, $5\%$ of the covariate $x_2$ is contaminated with distribution $U[1, 2]$, and in the same observations, the response is set to $y = 0$. Additionally, $5\%$ of the response variable is contaminated by adding a constant value $y_0 = 70$.

\subsection{Measure of Performance}

Various performance metrics are calculated to evaluate the obtained estimators. Specifically, to quantify the error of $\widehat{m}$, the well-known root mean square error (RMSE)
\begin{align}\label{rmse}
    RMSE(\widehat{m}) &= \left(\frac{1}{n} \sum_{i=1}^n(\widehat{m}_n(t_i) - m_0(t_i))^2\right)^{1/2}
\end{align}
is computed for each sample. 

Concerning the linear parameters, the distances between the estimators obtained with each sample and the true parameter value are calculated. That is $||\widehat{\b{\beta}}_n - \b{\beta}_0||$ and $||\widehat{\b{\gamma}}_n - \b{\gamma}_0||$, where $||\cdot||$ represents the $L2$-norm. 

\subsection{Results}

In this subsection, we present the results corresponding to  each contamination scheme. The estimators proposed in this work require the specification of the smoothing parameter $h_n$, associated with the Nadaraya-Watson kernel weights used in \eqref{Q_n1}. The value of $h_n$ was computed for each estimator using 5-fold cross-validation on the first ten simulated datasets. To select $h_n$, we applied 5-fold cross-validation to the first ten simulated datasets, following the procedure described in \cite{Boente06}. Given the similarity of the optimal $h_n$ values obtained across these datasets and in order to reduce the overall computational cost, we adopted the average of the ten selected values as a fixed $h_n$ for the remaining datasets. The final values of the smoothing parameter used for each estimator were $h_{n,ML} = 0.126$, $h_{n,MT}=0.135$ and $h_{n,CH}=0.159$ for ML, MT and CH estimators, respectively.

\

\begin{figure}
\begin{center}
    \includegraphics[width=1\textwidth]{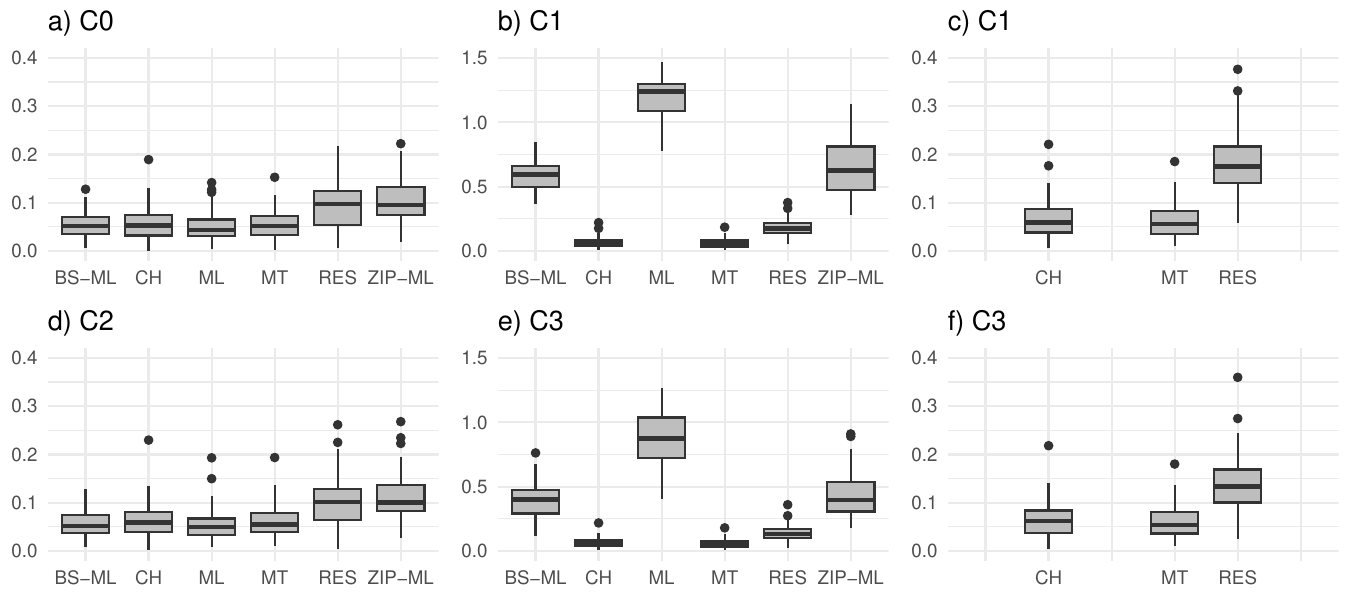}
    \caption{$L_2$-norm distance from $\b{\beta}_0$ to the estimators. The panels a), b), d) and e) represent each contamination scheme. c) and f) are zoomed views of the panels b) and e), respectively.}
    \label{betas}
\end{center}
\end{figure}

Figure \ref{betas} shows the $L_2$-norm distance between each $\b{\beta}$ estimator and the true parameter value $\b{\beta}_0$, for each method and contamination scheme. In Figure \ref{betas}a), it can be seen that under scheme $C_0$, the ML estimator proposed in this work shows the best performance, followed by the BS-ML estimator and the robust CH and MT estimators. The behavior of the robust estimators does not differ significantly from that of the ML estimators in this scenario, demonstrating their efficiency. In contrast, the RES and ZIP-ML estimators perform worse since the data were generated under a partially linear model. On the other hand, Figures \ref{betas}b) and \ref{betas}e) show that in schemes $C_1$ and $C_3$, where contamination is present in the response variable, the maximum likelihood-based estimators perform poorly, as expected. In this case the BS-ML and ZIP-ML estimators are more conservative than the ML estimator, since the latter offers more flexibility provided by the nonparametric component estimated using the weighted log-likelihood function, allowing the estimator to better adapt to the data. However, when the data contain extreme values, this flexibility becomes counterproductive. Figure \ref{betas}d) shows that scheme $C_2$ does not seem to affect the estimation of $\b{\beta}$, as the behavior of the estimators is very similar to scheme $C_0$. This is reasonable since high-leverage observations in the covariate $\mathbf{x}$ are controlled by the function $\omega_1$, and false zeros mainly affect the estimation of $\b{\gamma}$. Figures \ref{betas}c) and \ref{betas}f) provide a zoomed view of panels b) and e), respectively, allowing for a clearer examination of the behavior of the robust estimators. It can be seen that in both schemes $C_1$ and $C_3$, the RES estimator performs worse than the other robust methods, while the MT estimator shows the best performance, slightly outperforming CH. Overall, both MT and CH estimators demonstrate good performance on uncontaminated data and strong resistance to outliers on contaminated data.

\begin{figure}
\begin{center}
     \includegraphics[width=0.66\textwidth]{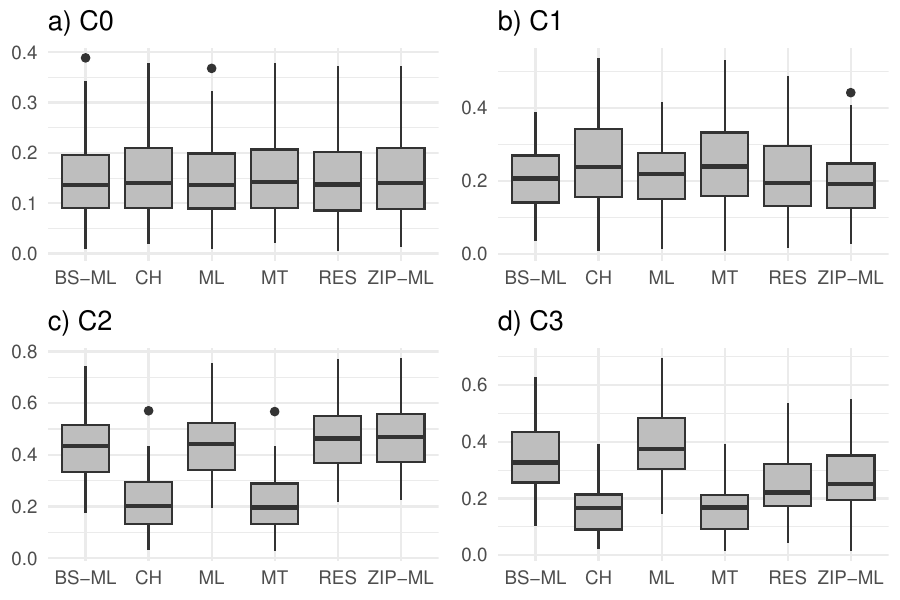}
    \caption{$L_2$-norm distance from $\b{\gamma}_0$ to each estimator. The panels a), b), c) and d) represent each contamination scheme.}
    \label{gammas}
\end{center}
\end{figure}
In Figure \ref{gammas} the $L_2$-norm distances between the $\b{\gamma}$ estimators and the true parameter value $\b{\gamma}_0$ are plotted. Specifically, in Figure \ref{gammas}a), it can be observed that all estimators achieve similar performance. In this case, the RES and ZIP-ML estimators are not affected by the specification of the model because it does not include a nonparametric component in the $\pi_i$ parameter. Figure \ref{gammas}b) shows that contamination under scheme $C_1$ induces a slight inefficiency in the robust MT, CH, and RES estimators. In the scenarios with contamination in the predictor variable, $C_2$ and $C_3$, the robust MT and CH estimators clearly outperform the other methods. This result is consistent, since the weights used to control extreme values of the covariates prevent false zeros from affecting the MT and CH estimators. The RES estimator performs poorly under this type of contamination, possibly due to the lack of weights that control extreme values on $\mathbf{x}$.

\begin{figure}
\begin{center}
    \includegraphics[width=\textwidth]{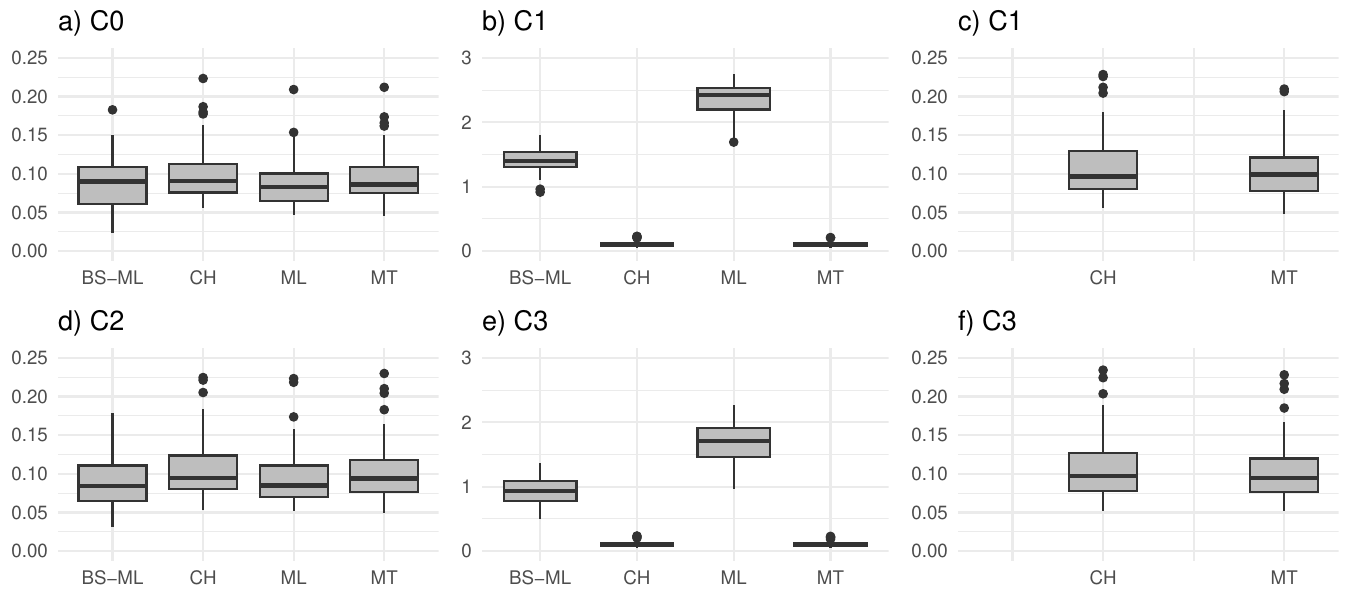}
    \caption{RMSE of each estimator of the nonparametric component $m_0$. The panels a), b), d) and e) represent each contamination scheme. c) and f) are zoomed views of the panels b) and e), respectively.}
    \label{m}
\end{center}
\end{figure}
Figure \ref{m} plots the RMSE values for the estimators of the function $m$, for those methods that allow its estimation. In scheme $C_0$, the BS-ML estimator performs slightly better than the others. In the remaining schemes, the behavior of the estimators is quite similar to that of the $\b{\beta}_0$ estimators described in Figure \ref{betas}, as the $\b{\beta}$ parameter and the function $m$ are estimated by minimizing the same objective function and, therefore, contaminated data affect both in a similar way. Specifically, we observe that under scheme $C_2$, the performance of the estimators is similar to that under $C_0$, while in schemes $C_1$ and $C_3$, the non-robust estimators are affected by extreme values. The performance of the robust estimators can be better appreciated in Figures \ref{m}c) and \ref{m}f), which provide zoomed views of panels b) and e), respectively. These plots show that the MT estimator achieves slightly better performance than the CH estimator.

\begin{figure}
\begin{center}
    \includegraphics[width=\textwidth]{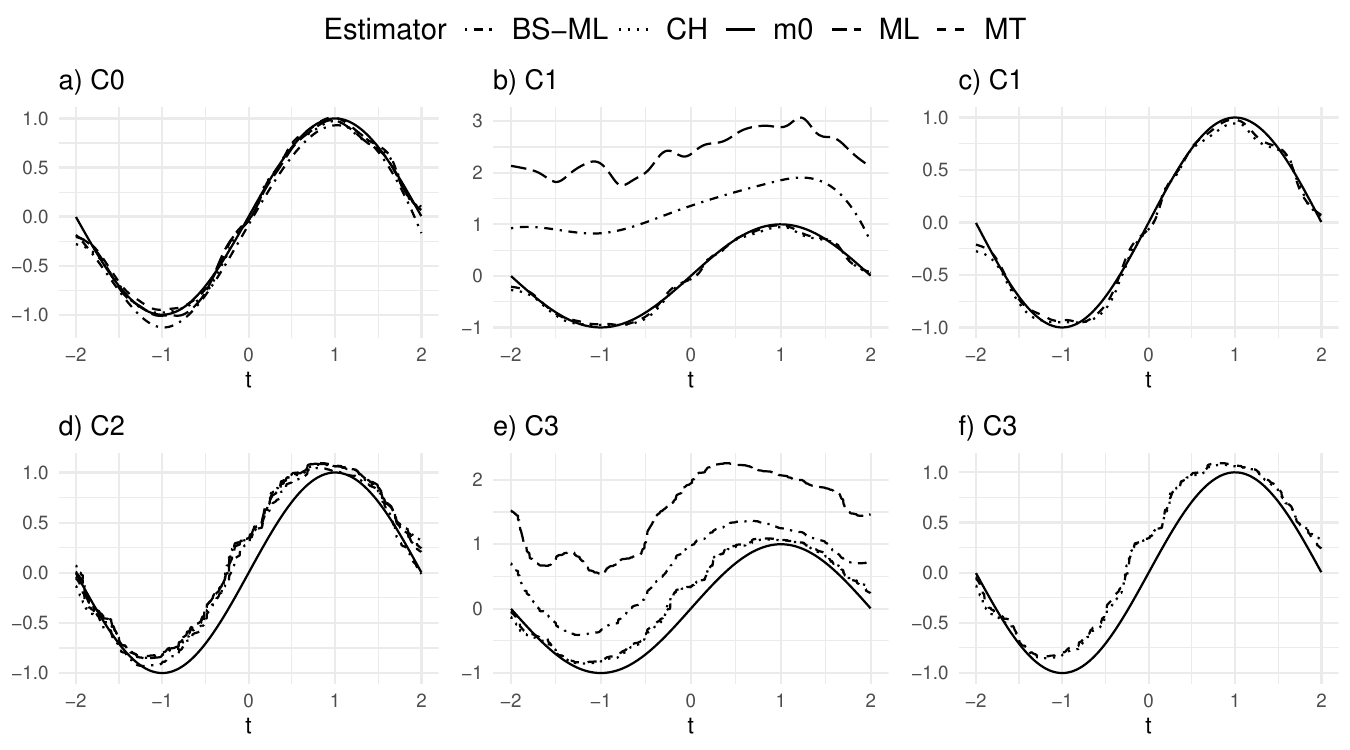}
    \caption{True function $m_0$ and its estimates obtained from a specific simulated dataset under each contamination scheme.}
    \label{ms}
\end{center}
\end{figure}
Figure \ref{ms} illustrates an example of the estimated function $m$ for a particular sample. It can be observed that under schemes $C_1$ and $C_3$, the likelihood-based estimators, ML and BS-ML, perform poorly, whereas the robust estimators maintain good performance. In contrast, under schemes $C_0$ and $C_2$, all estimators exhibit similarly good behavior.

\section{Example}\label{real_data}

In this section, a real-world data example is presented to evaluate the predictive power of the estimators. The example is based on data from a public health survey conducted in Indonesia in 1997, for more details see \cite{Frankenberg00}. The response variable of interest is the number of days that working-age ($18-60$ years) individuals missed their main activities due to illness in the previous four weeks self-reported by the respondent. The explanatory variables are household annual per capita income ($x_1$ in thousands), gender ($x_2 = 1$ for women and $0$ otherwise),  the household hygiene index ($x_3$ ranges from $0$ to $5$, from best to worst), and age ($t$ ranges from $18$ to $60$ years). One observation is randomly selected from each household to avoid the problem of possible intra-household correlation. We suppose that the population is divide into the low-risk group, which generates structural zeros; and the high-risk group, whose counts can be modeled using a Poisson regression model. We assumed that the proportion of low-risk subjects could be affected by the predictor variables $\z = (1, \x^\top, t)^\top$, where $\x^\top = (x_1,x_2,x_3)^\top$.

The number of days that individuals missed their main activities due to illness in the previous four weeks ranges from 0 to 28 with a mean of $1.05$ and variance of $10.16$. The variance is much larger than the mean, indicating potential overdispersion. As shown in Figure \ref{barplot}, the data contain more zero visits than a typical
Poisson distribution. Specifically, the percentage of the zero count is $6359/5073 \approx 80\%$. In addition, there are some particularly extreme counts.
\begin{figure}
\begin{center}
    \includegraphics[width=0.7\textwidth]{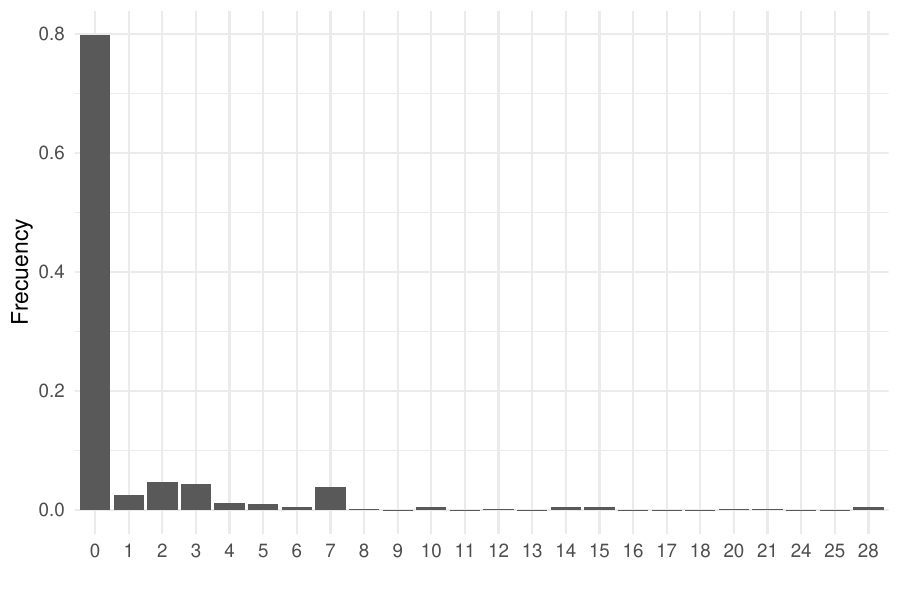}
    \caption{Barplot of the number of days of missed primary activities due to illness in the past 4 weeks self-reported by the respondent.}
    \label{barplot}
\end{center}
\end{figure}

This data were also analyzed by \cite{Lam06} using a semiparametric ZIP regression model. The authors mention that the age variable is generally not significant when assuming a linear relationship in a parametric setting, primarily because the effects of age can be very different across age groups. Therefore, they postulate that the age variable has a nonlinear relationship with the natural logarithm of the average number of days of absence from main activities due to illness among subjects in the high-risk group. In this article, the PLZIP model is employed to study the relationship between the number of days of absence and covariates. Based on the assumption of \cite{Lam06} we include the variable age in the nonparametric component, so that
\begin{align}
    \logit(\pi)&=\gamma_0 + \gamma_1 x_1 + \gamma_2 x_2 + \gamma_3 x_3 + \gamma_4 t,\\
    \log(\lambda)&=\beta_1 x_1 + \beta_2 x_2 + \beta_3 x_3 +m(t).
\end{align}
We estimate the model using the estimators proposed in this article, as well as those previously presented in Section \ref{simulations}. To provide a measure of the prediction error, we performed a 5-fold cross-validation and calculated the trimmed mean of the $20\%$ squared prediction error within each fold.

The results are presented in Figure \ref{error_data}. This indicates that predictions made using the maximum likelihood function (ZIP-ML, BS-ML, and ML) exhibit greater errors than those obtained with models estimated robustly (RES, MT, and CH), suggesting the presence of outliers, that is, a few observations that deviate from the model. 

Between the likelihood-based estimators, ZIP-ML, which does not include a non-parametric component, shows the worse performance. The ML estimator is overly adaptive to the data, due to the flexibility offered by kernel-based estimators, which can be a problem in the presence of extreme values in the data. The BS-ML estimator, which estimate the nonparametric component via B-splines, seems to be a slightly better option.

On the other hand, between the robust estimators, the predictions obtained using the PLZIP model with the MT and CH estimators proposed in this article, present smaller errors than the predictions made using the ZIP model with RES estimators. This suggests that, in addition to preferring the use of robust estimators, nonparametric modeling of the age variable is more appropriate than linear modeling.

\begin{figure}
\begin{center}
    \includegraphics[width=0.6\textwidth]{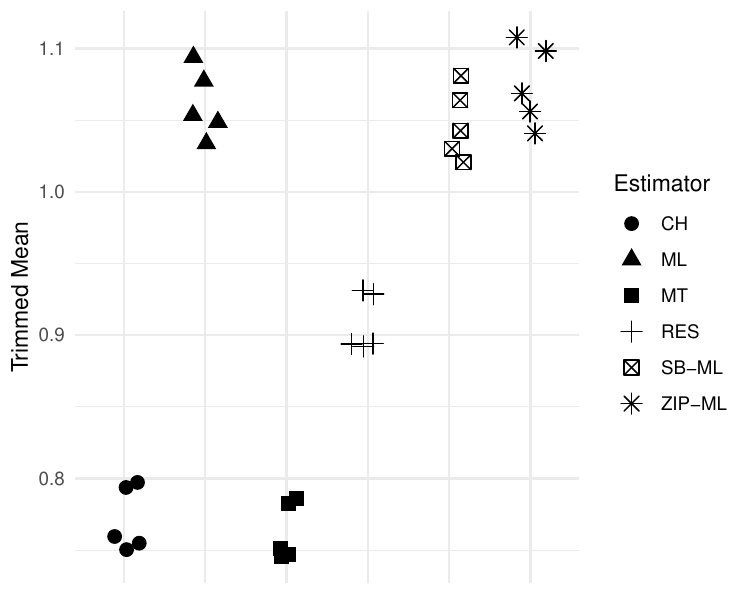}
    \caption{Trimmed mean of the quadratic error for each estimator in each fold.}
    \label{error_data}
\end{center}
\end{figure}

Table \ref{estimates} reports the estimates of the regression parameters $\b{\beta}$ and $\b{\gamma}$ obtained under the different estimation procedures. Several remarks can be made. First, the signs and magnitudes of the coefficients are generally consistent across methods. Regarding the estimates of $\b{\beta}$, income ($x_1$) shows a negative effect, indicating that higher socioeconomic status is associated with fewer absent days due to illness. Gender ($x_2$) generally shows a positive effect (except for the ZIP-ML and ML estimators), suggesting that women tend to report more absents due to illness. The hygiene index ($x_3$) also has a positive coefficient, indicating that poorer hygiene conditions increase the absents. On the other hand, the estimators of $\b{\gamma}$ show similar patterns across all methods: higher income ($x_1$) increases the probability of belonging to the low-risk group, while being female ($x_2 = 1$), poor household hygiene ($x_3$), and age ($t$) decrease it.

Figure \ref{estimates_m} shows the estimates of $m(t)$ for each partially linear method. It can be seen that the likelihood-based estimates fluctuate but generally show an increasing trend in the effect of age on the response. While the robust estimators show that the effect of age decreases until age 30 and then increases, remaining approximately constant between 50 and 60 years. This corroborates the fact that age does not have a simple linear effect on the response and suggests including $m(t)$ as a nonparametric component in the PLZIP specification. This allows the age effect to vary flexibly across groups, explaining the better performance of the robust partially linear estimators observed in the cross-validation study.

\begin{table}[ht]
\centering
\begin{tabular}{lcccccc}
  \hline
 & ZIP-ML & ML & SB-ML & RES & MT & CH \\ 
  \hline
$\beta_0$ & 1.633 & - & - & 1.393 & - & - \\ 
  $\beta_1$ & -0.017 & -0.008 & -0.017 & -0.018 & -0.27 & -0.302 \\ 
  $\beta_2$ & -0.009 & -0.002 & 0.05 & 0.048 & 0.167 & 0.176 \\ 
  $\beta_3$ & 0.013 & 0.015 & 0.016 & 0.019 & 0.003 & 0.006 \\ 
  $\beta_4$ & 0.098 & - & - & 0.084 & - & - \\ 
   \hline
$\gamma_0$ & 1.606 & 1.606 & 1.604 & 1.592 & 1.546 & 1.549 \\ 
  $\gamma_1$ & 0.04 & 0.041 & 0.04 & 0.04 & 0.075 & 0.075 \\ 
  $\gamma_2$ & -0.396 & -0.396 & -0.394 & -0.393 & -0.371 & -0.372 \\ 
  $\gamma_3$ & -0.048 & -0.048 & -0.048 & -0.047 & -0.048 & -0.047 \\ 
  $\gamma_4$ & -0.109 & -0.11 & -0.11 & -0.106 & -0.098 & -0.096 \\ 
   \hline
\end{tabular}
\label{estimates}
\caption{$\b{\beta}$ and $\b{\gamma}$ estimates.} 
\end{table}

\begin{figure}
\begin{center}
    \includegraphics[width=0.9\textwidth]{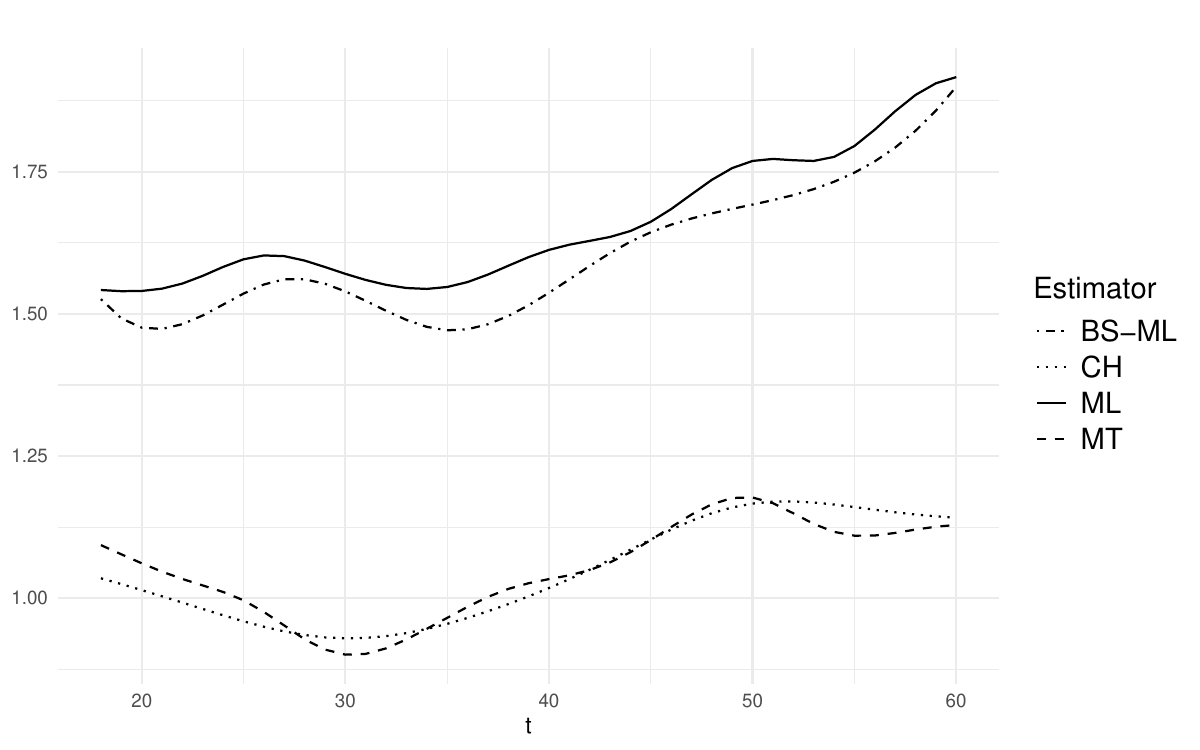}
    \caption{$m(t)$ estimates using public health survey data.}
    \label{estimates_m}
\end{center}
\end{figure}

This preliminary analysis shows that the use of partially linear models in a context of zero-inflated data, combined with robust estimation, becomes a powerful predictive tool in the face of the possible presence of extreme values and complex relationships between variables.

\section{Conclusions}

This study introduced an estimation procedure for the PLZIP regression model that allows obtaining robust estimators. 

The ML estimators of the parameters $\b{\beta}$ and the nonparametric function $m$ are susceptible to the presence of outliers in the response variable. Similarly, the ML estimator of the parameter $\b{\gamma}$ exhibits notable sensitivity to false zeros and extreme values in the covariates, resulting in poor performance under contamination. In contrast, the robust CH and MT estimators demonstrated  resistance to contamination in both the response and covariates. Both estimators showed similar performance, with the MT estimator slightly outperforming the CH in all simulation schemes. In non-contamination schemes, both robust estimators maintained comparable performance to the likelihood-based estimators.

The application of robust estimators in the analysis of real-world public health data reaffirmed their predictive superiority over those based on the likelihood. The robustly estimated models exhibited lower prediction errors, suggesting the presence of outliers or slight deviations from model assumptions in the data. Furthermore, the results indicated that nonparametric age modeling is more appropriate than a linear approach, highlighting the importance of the PLZIP model structure in capturing complex and nonlinear relationships.

It is worth mentioning that the estimation procedure developed in this article can be naturally extended to the doubly semiparametric framework. In this setting, the logistic link function incorporates a nonparametric component, relaxing the parametric assumptions and allowing for greater modeling flexibility. The implementation of this extension only requires adapting Step 1 of the procedure in the same way as for the estimation of $\b{\beta}$ and $m$. This modification results in a more versatile tool, applicable to a larger class of problems in complex data analysis scenarios.

In conclusion, the robust estimation methodology for the PLZIP model introduced in this work not only addresses the limitations of likelihood-based methods in the presence of outliers but also offers a powerful and versatile predictive tool. Therefore, the CH and MT estimators effectively combine robustness and efficiency, so their application in studies involving PLZIP models is highly recommended.

\section*{Acknowledgements}

This paper has been supported by the MATH-AmSud international grant 22MATH-07. We appreciate the support and funding provided by this initiative, which has been crucial in completing this work.

\appendix 
\section{}\label{appendix}

\textbf{Proof of Lemma \ref{lemma}  (a).}

Define
    \begin{align}
        R_{n,1}(\eta,\b{\beta},\tau|\b{\theta}) &= (nh_n)^{-1} \sum_{i=1}^n (1-E_\b{\theta}(w_i|y_i))\rho(y_i, \b{x_i}^\top\b{\beta} + \eta)\omega_1(\x_i) K((t_i-\tau)/h_n), \\
        R_{n,0}(\tau) &= (nh_n)^{-1} \sum_{j=1}^n K\left(\frac{t_j-\tau}{h_n}\right)
    \end{align}
and $Q_{n,1}(\eta,\b{\beta},\tau|\b{\theta})=R_{n,1}(\eta,\b{\beta},\tau|\b{\theta})/R_{n,0}(\tau)$. Then,     
    \begin{align}
        &\sup_{\b{\beta} \in \mathcal{K}\
        \eta \in \mathbb{R}} ||Q_{n,1}(\eta,\b{\beta},\cdot|\b{\theta}) - Q_1(\eta,\b{\beta},\cdot|\b{\theta})||_{0,\infty} = \\
        & \sup_{\b{\beta} \in \mathcal{K},\
        \eta \in \mathbb{R}} ||R_{n,1}(\eta,\b{\beta},\cdot|\b{\theta}) - Q_1(\eta,\b{\beta},\cdot|\b{\theta})R_{n,0}(\cdot)||_{0,\infty} \left(  \inf_{\tau \in \mathcal{T}_0} R_{n,0}(\tau)\right)^{-1} \leq \\
        & \Big( \sup_{\b{\beta} \in \mathcal{K}, \ 
       \eta \in \mathbb{R}} ||R_{n,1}(\eta,\b{\beta}, \cdot|\b{\theta}) - E_0(R_{n,1}(\eta,\b{\beta}, \cdot|\b{\theta}))||_{0,\infty} \\
        & + \sup_{\b{\beta} \in \mathcal{K},\
        \eta \in \mathbb{R}} ||E_0(R_{n,1}(\eta,\b{\beta}, \cdot|\b{\theta})) - Q_1(\eta,\b{\beta}, \cdot|\b{\theta})E_0(R_{n,0}(\cdot))||_{0,\infty}\\ 
        & \sup_{\b{\beta} \in \mathcal{K},\
        \eta \in \mathbb{R}} ||Q_1(\eta,\b{\beta}, \cdot|\b{\theta})||_{0,\infty} ||R_{n,0}(\cdot) -E_0(R_{n,0}(\cdot)) ||_{0,\infty}\Big) \left(  \inf_{\tau \in \mathcal{T}_0} R_{n,0}(\tau)\right)^{-1} .
    \end{align}
    C4 implies that $E_0(R_0(t)) = \int K(u) f_T(t-uh_n)du>A_1(\mathcal{T}_\delta)$. Therefore it is enough to prove the following convergences:
    \begin{enumerate}[(A)]
    \item $\sup_{\b{\beta} \in \mathcal{K},\ 
        \eta \in \mathbb{R}} ||R_{n,1}(\eta,\b{\beta}, \cdot|\b{\theta}) - E_0(R_{n,1}(\eta,\b{\beta}, \cdot|\b{\theta}))||_{0,\infty} \cs 0$,
    \item $\sup_{\b{\beta} \in \mathcal{K},\
        \eta \in \mathbb{R}} ||E_0(R_{n,1}(\eta,\b{\beta}, \cdot|\b{\theta})) - Q_1(\eta,\b{\beta}, \cdot|\b{\theta})E_0(R_{n,0}(\cdot))||_{0,\infty} \cs 0 $,
    \item $\sup_{\b{\beta} \in \mathcal{K},\
        \eta \in \mathbb{R}} ||Q_1(\eta,\b{\beta}, \cdot|\b{\theta})||_{0,\infty} ||R_{n,0}(\cdot) -E_0(R_{n,0}(\cdot)) ||_{0,\infty} \cs 0 $.
    \end{enumerate}
To prove (A), Theorem 37 of \cite{Pollard84} is applied to the family of functions
    \begin{align}
        \mathcal{F}_n = \Big\lbrace f_{\eta,\b{\beta},\tau}^n (y,\x,t|\b{\theta}) = \frac{(1-E_\b{\theta}(w|y))\rho(y,\x^\top \b{\beta}+\eta)\omega_1(\x)K((t-\tau)/h_n)}{||\rho||_\infty||\omega_1||_\infty||K||_\infty}, \\
        \b{\beta} \in \mathcal{K}, \ \eta \in \mathbb{R}, \ \tau \in \mathcal{T}_0 \Big\rbrace,
    \end{align}
    where $||\rho||_\infty = \sup_{y,u} |\rho(y,u)|$, $||\omega_1||_\infty = \sup_\x |\omega_1(\x)|$ y $||K||_\infty = \sup_{v} |K(v)|$. Note that $\mathcal{F}_n $ depends on $n$ through $h_n$. To ensure the conditions of the theorem, consider the family
    \begin{align}
        \mathcal{K}_n = \left\lbrace k_\tau^n(t) =  K\left(\frac{t-\tau}{h_n}\right), \tau \in \mathcal{T}_0 \right\rbrace,
    \end{align}    
    where $K$ is a bounded variation function on $\mathbb{R}$ satisfying $K\leq||K||_\infty$. Then, by Lemma 22 (ii) of \cite{Nolan87}, the graphs of functions in $\mathcal{K}_n$ have polynomial discrimination. Therefore, by Lemma 25 of \cite{Pollard84}, there are real values $A_1$ and $W_1$ (depending only on the discriminant polynomial) such that for any $0<\varepsilon<1$ and probability measure $Q$, $\mathcal{N}(\varepsilon||K||_\infty, \mathcal{K}_n, L^1(Q)) \leq A_1 \varepsilon^{-W_1}$. So, there are constants $A_2$ and $W_2$ such that $\sup_Q \mathcal{N}(\mathcal{F}_n||K||_\infty,\varepsilon, L^1(Q))<A_2\varepsilon^{-W_2}$ for all $0<\varepsilon<1$. In addition, since $|1-E_\b{\theta}(w_i|y_i)| < 1$, $|f_{\eta,\b{\beta},\tau}^n|<1$ for all $f_{\eta,\b{\beta},\tau}^n \in \mathcal{F}_n$. Furthermore,
    \begin{align}
        E_0((f_{\eta,\b{\beta},\tau}^n (y,\x,t|\b{\theta})))^2 &\leq E_0(f_{\eta,\b{\beta},\tau}^n (y,\x,t|\b{\theta}))) \\
        &= E_0\left( \frac{(1-E_\b{\theta}(w|y))\rho(y,\x^\top \b{\beta}+\eta)\omega_1(\x)K((t-\tau)/h_n)}{||\rho||_\infty||\omega_1||_\infty||K||_\infty} \right) \\
        &\leq E_0\left( \frac{K((t-\tau)/h_n)}{||K||_\infty} \right) \\
        &= \frac{1}{||K||_\infty} \int K\left(\frac{t-\tau}{h_n}\right) f_T(t)dt \\
        (l = (t-\tau)/h_n) &= \frac{h_n}{||K||_\infty} \int K(l) f_T(lh_n+\tau)dl \\
        &\leq \frac{h_n ||f_T||_\infty}{||K||_\infty} \int K(l)dl \\
        &= h_n ||f_T||_\infty||K||_\infty^{-1} \doteq \delta_n^2.
    \end{align}
    Taking $\alpha_n=1$, $n\delta_n^2 \alpha_n^2 >> \log (n)$, 
    \begin{align}
        \sup_{\mathcal{F}_n} \Big|\Big|\frac{1}{n} \sum_{i=1}^n f_{\eta,\b{\beta},\tau}^n (y_i,\x_i,t_i|\b{\theta})) - E_0(f_{\eta,\b{\beta},\tau}^n (y,\x,t|\b{\theta})))\Big|\Big| << \delta_n^2 \alpha_n \quad \text{c.s.},
    \end{align}
    it means that
    \begin{align}
        \frac{1}{ ||f_T||_\infty||\rho||_\infty||\omega_1||_\infty}  \sup_{\b{\beta} \in \mathcal{K},\ \eta \in \mathbb{R}} ||R_{n,1}(\eta,\b{\beta}, \cdot|\b{\theta})) - E_0(R_{n,1}(\eta,\b{\beta}, \cdot|\b{\theta})))||_{0,\infty} \cs 0,
    \end{align} 
    and the first factor is finite, so (A) is proven.

    To prove (B), note that 
    \begin{align}
       & |E_0(R_{n,1}(\eta,\b{\beta}, \tau|\b{\theta}))) - Q_1(\eta,\b{\beta}, \tau|\b{\theta}))E_0(R_{n,0}(\tau))| \\ 
        & = h_n^{-1} \Big| \int Q_1(\eta,\b{\beta},t|\b{\theta})) K\left(\frac{t-\tau}{h_n}\right)f_T(t)dt \\
        & - \int Q_1(\eta,\b{\beta},\tau|\b{\theta})) K\left(\frac{t-\tau}{h_n}\right)f_T(t)dt\Big|\\
        & = \Big|\int (Q_1(\eta,\b{\beta},lh_n+\tau|\b{\theta}))-Q_1(\eta,\b{\beta},\tau|\b{\theta}))) K(l)f_T(lh_n+\tau)dl\Big|\\
        & \leq ||f_T||_\infty \int \Big|Q_1(\eta,\b{\beta},lh_n+\tau|\b{\theta}))-Q_1(\eta,\b{\beta},\tau|\b{\theta}))\Big| K(l)dl.
    \end{align}
    The fact that $h_n \rightarrow 0$ and the condition C5 imply that
    \begin{equation}
        \sup_{\b{\beta} \in \mathcal{K},\
        \eta \in \mathbb{R}} ||E_0(R_{n,1}(\eta,\b{\beta}, \cdot|\b{\theta}))) - Q_1(\eta,\b{\beta}, \cdot|\b{\theta}))E_0(R_{n,0}(\cdot))||_{0,\infty} \rightarrow 0.
    \end{equation}
    To prove (C) it is enough to show that 
    \begin{align}\label{c3}
        ||R_{n,0}(\cdot) -E_0(R_{n,0}(\cdot)) ||_{0,\infty} \cs 0.
    \end{align}
    Since $|1-E_\b{\theta}(w|y)|<1$,
    \begin{equation}
        \sup_{\b{\beta} \in \mathcal{K},\ \eta \in \mathbb{R}} ||Q_1(\eta,\b{\beta}, \cdot| \b{\theta})||_{0,\infty} 
    \leq ||\rho||_\infty ||\omega_1||_\infty,
    \end{equation} 
    with $\rho$ and $\omega_1$ bounded. 
    
    The hypothesis of Theorem 37 of \cite{Pollard84} are verified for the family $\mathcal{K}_n$. Indeed, without loss of generality, if the kernel function is bounded by $1$
    \begin{align}
        \int K((t-\tau)/h_n)^2 f_T(t)dt &\leq \int K((t-\tau)/h_n) f_T(t)dt \\ 
        &= h_n \int K(l) f_T(lh_n+\tau)dl \\
        &\leq h_n ||f_T||_\infty,
    \end{align}
    and taking $\alpha_n= 1$, by hypothesis $n\delta_n^2\alpha_n^2 >> \log (n)$. Thus,
    \begin{align}
    \sup_{\mathcal{K}_n} \left| \frac{1}{n}\sum_{i=1}^n K((t_i-\tau)/h_n) -  E_0( K((t-\tau)/h_n)) \right| << ||f_T||_\infty h_n \quad a.s.,
    \end{align}
    and then
    \begin{align}
    \sup_{\mathcal{K}_n}|R_{n,0}(\tau) -E_0(R_{n,0}(\tau))| \cs 0,
    \end{align}
    which concluded the proof of (a).
$\hfill\blacksquare$

\

\textbf{Proof of Theorem \ref{theorem_beta}. } 
To prove (a), note that
    \begin{align}
        \sup_{\b{\beta} \in \mathbb{R}^{p}} |Q_{n,2}(\Tilde{m}, \b{\beta}| \b{\theta})-Q_2(m_0,\b{\beta}| \b{\theta})| &\leq
        \sup_{\b{\beta} \in \mathbb{R}^{p}} |Q_{n,2}(\Tilde{m},\b{\beta}| \b{\theta})-Q_{n,2}(m_0,\b{\beta}| \b{\theta})| \\
        &+ \sup_{\b{\beta} \in \mathbb{R}^{p}} |Q_{n,2}( m_0,\b{\beta}| \b{\theta})-Q_2( m_0,\b{\beta}| \b{\theta})|.
    \end{align}
    So, it is enough to show the convergences
    \begin{enumerate}[(A)]
    \item $\sup_{\b{\beta} \in \mathbb{R}^{p}} |Q_{n,2}(\Tilde{m},\b{\beta}| \b{\theta})-Q_{n,2}(m_0,\b{\beta}| \b{\theta})| \cs 0$, 
    \item $\sup_{\b{\beta} \in \mathbb{R}^{p}} |Q_{n,2}(m_0,\b{\beta}| \b{\theta})-Q_2(m_0,\b{\beta}| \b{\theta})| \cs 0$.
    \end{enumerate}
    Let $\mathcal{T}_0\in \mathcal{T}$ be a compact set such that $P(t_i \notin \mathcal{T}_0) < \varepsilon$, the first-order Taylor expansion of $\rho$ around $\x_i^\top\b{\beta} + m_0(t_i)$ for each $i=1,\ldots,n$ implies that
     \begin{align}
       & \sup_{\b{\beta} \in \mathbb{R}^{p}} |Q_{n,2}(\Tilde{m},\b{\beta}| \b{\theta})-Q_{n,2}(m_0,\b{\beta}| \b{\theta})|\\ 
        & \leq \sup_{\b{\beta} \in \mathbb{R}^{p}} \frac{1}{n}\sum_{i=1}^n |\Psi(y_i, u_i)||\Tilde{m}(t_i) - m_0(t_i)||\omega_1(\x_i)|\mathcal{I}(t_i \in \mathcal{T}_0) \\
        & + \sup_{\b{\beta} \in \mathbb{R}^{p}} \frac{1}{n}\sum_{i=1}^n \Big( |\rho(y_i, \x_i^\top\b{\beta} + \Tilde{m}(t_i))|  + |\rho(y_i, \x_i^\top\b{\beta} + m_0(t_i))| \Big) |\omega_1(\x_i)| \mathcal{I}(t_i \notin \mathcal{T}_0)\\
        & \leq ||\Psi||_{\infty} ||\Tilde{m}- m_0||_{0,\infty} ||\omega_2||_{\infty}  + 2||\rho||_{\infty}||\omega_2||_{\infty} \frac{1}{n}\sum_{i=1}^n \mathcal{I}(t_i \notin \mathcal{T}_0),
    \end{align} 
    where $u_i \in (\x_i^\top\b{\beta} + 
        \Tilde{m}(t_i), \x_i^\top\b{\beta} + m_0(t_i)) \ \forall \ i=1,\ldots, n$.

        Since $\Psi$ and $\omega_2$ are bounded, and $||\Tilde{m}- m_0||_{0,\infty} \cs 0$ in $\mathcal{T}_0$, the first term tends to 0 a.s.  Concerning the second term, since $\rho$ is bounded and $E(\mathcal{I}(t_i \notin \mathcal{T}_0)) = P(t_i \notin \mathcal{T}_0) < \varepsilon$ the Strong Law of Large Numbers implies that the second term tends to 0 a.s.

        On the other hand, the hypothesis over the covering number of $\mathcal{H}$, the fact that the functions in $\mathcal{H}$ have the uniform bound $||\rho||_\infty ||\omega_1||_\infty$ and Theorem 24 of \cite{Pollard84} ensure that (B) holds, which concludes the proof of (a).

        To prove (b), let $\widehat{\b{\beta}}_k$ be a convergent sequence of $\widehat{\b{\beta}}$ with limit $\b{\beta}^*$. Suppose without loss of generality that $\widehat{\b{\beta}} \cs \b{\beta}^*$. 
        
    If $||\b{\beta}^*||<\infty$, from (a) $Q_{n,2}(\widehat{m},\b{\beta}_0|\b{\theta}) - Q_2(m_0,\b{\beta}_0|\b{\theta}) \cs 0$. Moreover, from the continuity of $Q_2$, $Q_{n,2}(\Tilde{m},\widehat{\b{\beta}}|\b{\theta})-Q_2(m_0,\b{\beta}^*|\b{\theta}) \cs 0$, indeed
        \begin{align}
            |Q_{n,2}(\Tilde{m},\widehat{\b{\beta}}|\b{\theta})-Q_2(m_0,\b{\beta}^*|\b{\theta})| &\leq |Q_{n,2}(\Tilde{m},\widehat{\b{\beta}}|\b{\theta})-Q_2(m_0,\widehat{\b{\beta}}|\b{\theta})| \\
            &+ |Q_2(m_0,\widehat{\b{\beta}}|\b{\theta})-Q_2(m_0,\b{\beta}^*|\b{\theta})|.
        \end{align}
        The first term tends to 0 a.s. by (a). Regarding the second term, the first-order Taylor expansion of $\rho$ around $\x^\top \b{\beta}^* + m_0(t)$ implies that 
        \begin{align}
            |Q_2(m_0,\widehat{\b{\beta}}|\b{\theta})-Q_2(m_0,\b{\beta}^*|\b{\theta})| &\leq E_0\Big(|1-E_\b{\theta}(w|y)||\rho(y,\x^\top\widehat{\b{\beta}}+m_0(t))\\
            &-\rho(y,\x^\top\b{\beta}^*+m_0(t))||\omega_1(\x)|\Big) \\
            &\leq E_0\left(|\Psi(y,u)||\x^\top(\widehat{\b{\beta}} - \b{\beta}^*)||\omega_1(\x)|\right) \\
            &\leq ||\Psi||_\infty E_0\left(||\x^\top||\right)||\widehat{\b{\beta}} - \b{\beta}^*|| ||\omega_1||_\infty,
        \end{align} 
        where $u \in (\x^\top\widehat{\b{\beta}} + 
        \eta, \x^\top\b{\beta}^* + \eta)$, $\Psi$ and $\omega_1$ are bounded functions, $\x$ has finite second moment, and $\widehat{\b{\beta}} \cs \b{\beta}^*$. Therefore $Q_2(m_0,\widehat{\b{\beta}}|\b{\theta})-Q_2(m_0,\b{\beta}^*|\b{\theta}) \cs 0$, which implies $Q_{n,2}(\Tilde{m},\widehat{\b{\beta}}|\b{\theta})-Q_2(m_0,\b{\beta}^*|\b{\theta}) \cs 0$.

        Note that $Q_{n,2}(\Tilde{m},\widehat{\b{\beta}}|\b{\theta}) \leq Q_{n,2}(\Tilde{m},\b{\beta}_0|\b{\theta})$, since $\widehat{\b{\beta}} = \argmin_{\b{\beta} \in \mathbb{R}^p} Q_{n,2}(\Tilde{m},\b{\beta}|\b{\theta})$. Then, taking a.s. limit $Q_2(m_0,\b{\beta}^*|\b{\theta}) \leq Q_2(m_0,\b{\beta}_0|\b{\theta})$, but $(\b{\beta}_0, m_0)$ is the only minimum of $Q_2$, so it must be $\b{\beta}^*=\b{\beta}_0$.

        If $||\b{\beta}^*||=\infty$, by (a) and the continuity of $Q_2$, $Q_{n,2}(\Tilde{m},\b{\beta}_0|\b{\theta}) - Q_2(m_0,\b{\beta}_0|\b{\theta}) \cs 0$ and $Q_{n,2}( \Tilde{m},\widehat{\b{\beta}}|\b{\theta})-Q_2(m_0,\b{\beta}^*|\b{\theta}) \cs 0$. Moreover, $Q_{n,2}(\Tilde{m},\widehat{\b{\beta}}|\b{\theta}) \leq Q_{n,2}(\Tilde{m},\b{\beta}_0|\b{\theta})$. Hence
         \begin{align}
             0 \geq Q_{n,2}( \Tilde{m},\widehat{\b{\beta}}|\b{\theta}) -Q_{n,2}(\Tilde{m},\b{\beta}_0|\b{\theta}) &\geq Q_{n,2}( \Tilde{m},\widehat{\b{\beta}}|\b{\theta})-Q_2(m_0,\b{\beta}^*|\b{\theta}) \\
             &- (Q_{n,2}(\Tilde{m},\b{\beta}_0|\b{\theta}) - Q_2(m_0,\b{\beta}_0|\b{\theta})) \\
             &+  Q_2(m_0,\b{\beta}^*|\b{\theta}) - Q_2(m_0,\b{\beta}_0|\b{\theta}).
         \end{align}
        Taking a.s. limit, 
        \begin{align}
           \lim_{||\b{\beta}|| \rightarrow \infty} Q_2(m_0,\b{\beta}|\b{\theta}) \geq Q_2(m_0,\b{\beta}_0|\b{\theta}),
        \end{align}
        which, by hypothesis, is absurd. Therefore it must be $||\b{\beta}^*||<\infty$.
$\hfill\blacksquare$

\bibliographystyle{elsarticle-num} 
\bibliography{Bibs}

\end{document}